\documentclass[a4paper]{article}

\usepackage{amsmath}
\usepackage{algorithm}
\usepackage{amsfonts}
\usepackage{pgfplots}
\usepackage[colorinlistoftodos]{todonotes}
\usepackage[colorlinks=true, allcolors=blue]{hyperref}
\usepackage[noend]{algpseudocode}
\usepackage{amsthm}
\usepackage{cleveref}

\DeclareMathOperator*{\argmax}{\arg\!\max}

\begin{document}

\title{\LARGE \bf
Polynomial Time Algorithms for Constructing Optimal  Binary AIFV-$2$ Codes \thanks{This is an expanded version of work that originally appeared in \cite{golin2019polynomial}.}
}
\author{Mordecai  Golin$^{1}$ and Elfarouk Harb$^{2}$
\thanks{$^{1}$ Department of Computer Science \& Engineering, 
HKUST,  (golin@cs.ust.hk). Research partially supported by Hong Kong RGC GRF Grant 16213318.}
\thanks{$^{2}$ Departments of Mathematics and  Computer Science \& Engineering,  
   HKUST, (eyfmharb@connect.ust.hk)} }

\newdimen{\algindent}
\setlength\algindent{1.5em}          

\newtheorem{thm}{Theorem}
\newtheorem{lem}[thm]{Lemma}
\newtheorem{cor}[thm]{Corollary}
\newtheorem{defn}{Definition}
\newtheorem{example}{Example}
\newcommand{\q}[1]{``#1''} 

\newenvironment{Note}
{\em \small
Note: 
}

\newcommand{\Cost}{\mathrm{Cost}}
\newcommand{\Trunc}{\mathrm{Trunc}}
\newcommand{\sig}{\mathrm{sig}}
\newcommand{\depth}{\mathrm{depth}}
\newcommand{\Expand}{\mathrm{Expand}}
\newcommand{\OPT}{\mathrm{OPT}}
\newcommand{\Pred}{\mathrm{Pred}}
\newcommand{\CALI}{{\cal I}}
\newcommand{\bfx}{{\bf x}}
\newcommand{\bfz}{{\bf z}}
\newcommand{\bfy}{{\bf y}}
\newcommand{\bfq}{{\bf q}}
\newcommand{\bfp}{{\bf p}}
\tikzstyle{VertexStyle} = [shape=ellipse,minimum width= 6ex,draw]
\tikzstyle{EdgeStyle}   = [->,>=stealth']

\maketitle

\begin{abstract}
Huffman Codes are {\em optimal}  Instantaneous Fixed-to-Variable (FV) codes  in which  every source symbol can only be encoded by one codeword. Relaxing these constraints permits constructing better FV codes.
More specifically, recent work has shown that AIFV-$m$ codes can beat Huffman coding.  AIFV-$m$ codes construct am $m$-tuple of different coding trees between which the  code alternates  and  are only {\em almost instantaneous} (AI).  This means that decoding  a word might require a delay of a finite number of bits.
 
Current algorithms for constructing optimal AIFV-$m$ codes are iterative processes that construct progressively ``better sets'' of code trees.    
The  processes have been proven to finitely converge to the optimal code but with no known bounds  on the convergence rate.

This paper derives  a geometric interpretation  of the space of binary AIFV-$2$ codes, permitting the development of the first  polynomially time-bounded iterative procedures for constructing optimal AIFV   codes.
 
We first show that a simple binary search procedure can replace the current iterative process to construct  optimal binary AIFV-$2$ codes. We then 
describe how to frame the problem as  a linear programming with an exponential number of constraints but a polynomial-time  separability oracle.  This permits using the 
Gr{\"o}tschel,  Lov{\'a}sz and  Schrijver
ellipsoid method to solve the problem in a polynomial number of steps. While more complicated, this second method has the potential to lead to a polynomial time algorithm to construct optimal AIFV-$m$ codes for general $m$.

\medskip

\end{abstract}



\section{INTRODUCTION}

Although it is often loosely claimed  that Huffman coding  is {\em optimal} this is  only true within strictly defined conditions.  Huffman coding 
achieves optimal average code length among  all fixed to variable (FV) codes for stationary memoryless sources
 when only a single code tree is used. In \cite{original}, {\em $K$-ary Almost Instantaneous FV ($m$-AIFV)} codes were  introduced and shown to attain  better average code lengths than Huffman Coding in many cases.  It did this by using multiple code trees as well as dropping the Huffman coding assumption of instantaneous decoding and instead permitting  a bounded decoding delay.

An  AIFV code is a tuple of code trees with an associated encoding/decoding procedure (to be described in more detail in \Cref{sec:Intro}).

$K$-ary AIFV codes  used a $K$ character encoding alphabet;
for $K>2$, the procedure used $K-1$ coding trees and had  a coding delay of $1$ bit.  For $K=2$ it used 2 trees and had a coding delay of $2$ bits.

\cite{introduceMAIFV} later introduced {\em Binary} AIFV-$m$ codes.  These were binary codes with $m$ coding trees and  decoding delay of at most $m$ bits.  The Binary AIFV-$2$ codes of \cite{introduceMAIFV}  are identical to the $2$-ary AIFV codes of  \cite{original}. \cite{introduceMAIFV} proved that  Binary AIFV-$m$ codes have redundancy of at most $1/m.$

Constructing optimal AIFV codes is much more difficult than constructing  Huffman codes. 
\cite{yamamoto2015almost} described a  general iterative  technique. 
 \cite{iterativealgomaifv}  gave a finite  iterative procedure for constructing optimal 
{Binary} AIFV-$m$ codes and proved its correctness for  $m=2,3,4,5.$   The same algorithm was  later proven correct for $m>5$ by
\cite{iterativeisoptimal}.

The algorithm works by constantly refining  a tuple of  {\em cost parameters}  $C$ (corresponding to coding penalties).
 Each iterative step solves an Integer Linear Programming (ILP) problem that constructs the optimal code trees associated with  the current  values of $C$.  $C$  is  then updated to a new uple and the process is   repeated.   
 \cite{yamamoto2015almost}
  proved  that 
    this 
 converges in a finite number of steps to a fixed point parameter  $C^{*}$ and that 
 $C^{*}$'s  associated $m$-tuple of trees is an optimal binary AIFV-$m$ code.
 This procedure had two weaknesses.  The first was that it required solving NP-Hard ILPs at every step.  The second was that no bound was known on the number of steps required for the procedure to converge.
 
A  later paper,  \cite{dp}, replaced the ILP step for $m=2$  with 
 a  $O(n^5)$ time Dynamic Programming (DP) approach based on  techniques originally developed for constructing constrained prefix-free coding trees in  \cite{golin,chan2000dynamic}. 
  \cite{iterativealgomaifv} extended this approach to $m >2$ by replacing the ILPs there with DPs running in  $O(mn^{2m+1})$ time (under certain implicit assumptions).
 For all $m$, the state of the art was that the algorithms still required an unknown number of steps to converge.

The main result of this paper is two new iterative procedures for constructing binary AIFV-$2$ codes. Each iterative step of the new procedure still finds a pair tuple of code trees associated with  $C.$ 
 The main difference is that our new procedures are guaranteed to terminate in a polynomial number of steps, i.e., $O(b),$ where  $b$ is the maximum number of bits used to encode any one input probability $p_i.$
 
Our first procedure is a simple binary search.
Our second procedure follows from the observation that the problem can be recast as a Linear Program with exponentially many constraints but which has a linear separability oracle.  This permits solving it using    Gr{\"o}tschel,  Lov{\'a}sz and Schrijver's 
ellipsoid method \cite{ellipsoid}.  Our first procedure is specific to the $m=2$ case.  The second has the potential to be expanded to work for $m >2$ as well.

\begin{figure*}[h!]
 \center
\includegraphics[width=4.5in]{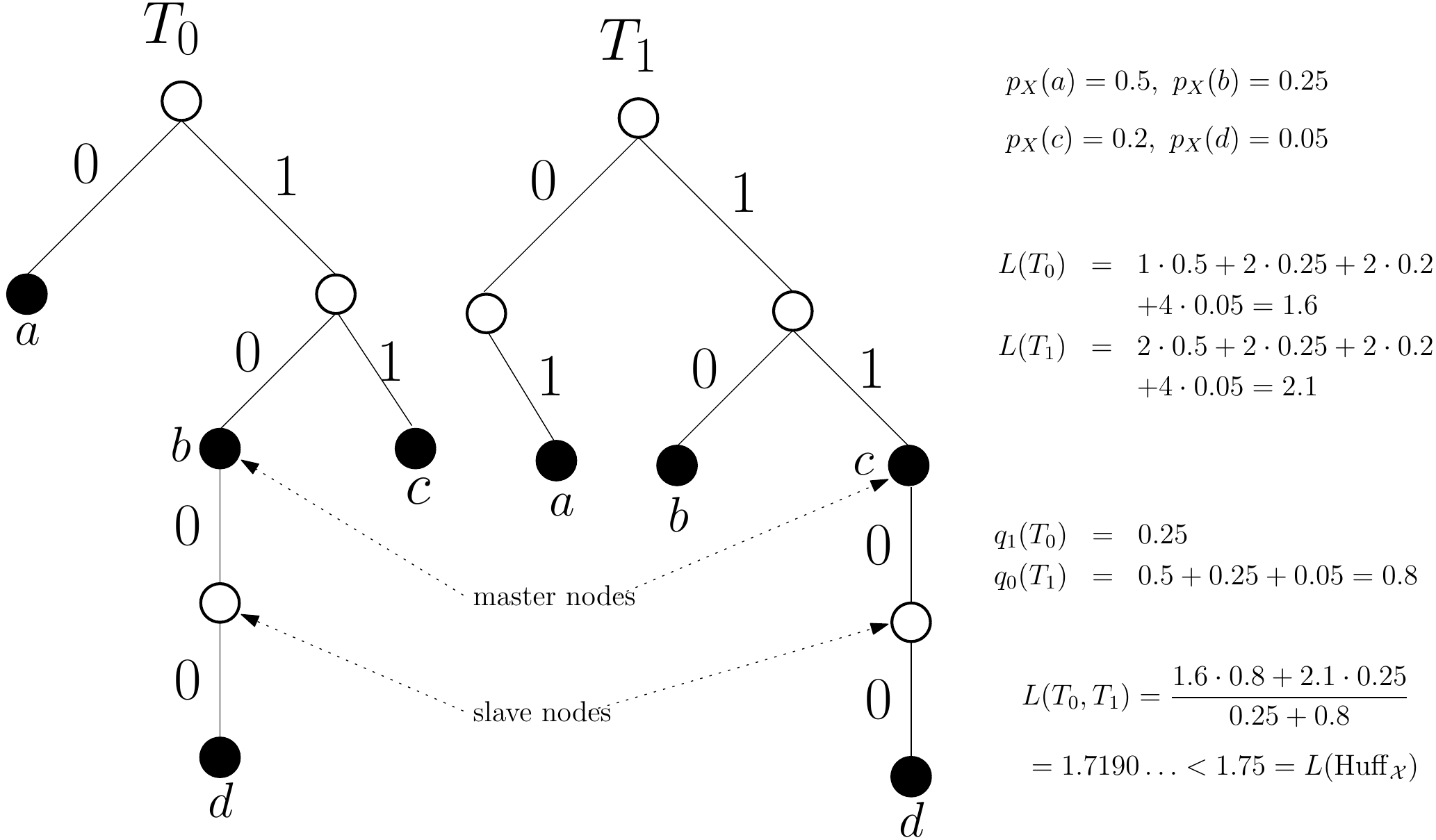}
  \caption{A binary AIFV-$2$ code for  $\mathcal{X}=\{a,b,c,d\}$ with associated probabilities. The encoding of  ${\bf b}d {\bf b}c {\bf a}a$ is
  ${\bf Y} = {\bf 10}1100 {\bf 10} 11 {\bf 01}0 .$ Note that $d,c$ and the first $a$  were encoded using $T_1$ while  the other letters were encoded using $T_0.$   $L({\mathrm{Huff}}_{\mathcal X})$ is the cost of the Huffman code for the same distribution.
  }
  \label{fig:BBB}\end{figure*}

Section \ref{sec:Intro}  introduces the previous work done on AIFV-$2$ codes. Section \ref{sec:LE and Geo}  provides a geometric reinterpretation  of \cite{yamamoto2015almost}'s iterative procedure for $m=2$
  as finding the intersection of two lower envelopes of lines.
Section \ref{sec: new algs}   then  uses this new interpretation to propose two  $O(n^5b)$ time algorithms for solving the same problem.  The first is essentially a simple binary search.  The second is a more complicated  ellipsoid-algorithm  separation-oracle based solution. 
  Section \ref{sec:Conc} briefly describes how this second technique might be extendable  to   $m >2$.

\section{Preliminaries and Previous Work on Binary AIFV-$2$ codes}

\label{sec:Intro}
Let $X$ be a memoryless  source over a finite alphabet $\mathcal{X}$ of size $n$.
$ \forall a_i \in \mathcal{X}$,   let $p_i = P_{X}(a_i)$ denote the probability of $a_i$ ocurring.
Without loss of generality  we assume that 
$$p_1 \geq p_2 \geq \cdots \geq p_n >0  \text{ and } \sum_{i=1}^{n}p_i = 1.$$
A  {\em codeword} $c$ of a binary AIFV code is  a string in $\{0,1\}^*.$
$|c|$ will denote   the length of codeword $c$. 

We assume that all probabilities are represented  using binary. Let $b_i$ be the number of bits used to represent $p_i$.
For later use we set $b= \max_i b_i.$ 
Since $\min_i p_i =  p_n \le \frac 1 n$, this implies $b = \Omega(\log n)$.  Note that that the total size of the problem's input is $\Omega(L)$ where 
$L = \sum_i b_i = \Omega(n \log n).$

We now briefly describe the structure of  Binary AIFV-$2$ codes using the terminology of \cite{introduceMAIFV}.
See  \cite{introduceMAIFV} for more details and \Cref{fig:BBB} for an example.
Codes are represented via binary trees with left edges  labelled by ``$0$'' and right edges by ``$1$''.
A code consists of  a pair of binary code trees,  $T_0, T_1$ satisfying: 
 \begin{itemize}
 \item Complete internal nodes in $T_0$ and $T_1$ have both left and right children.
\item Incomplete internal nodes (with the unique exception of the left child of the root of $T_1$)  have only a ``$0$'' (left) child.\\
Incomplete internal nodes are labelled as  either a {\em master} node or a {\em slave} node. 
\item  A master node  must be an incomplete node with an incomplete child\\
The child of a master node is a slave node.\\
{\em This implies that a master node is connected to its unique grandchild via  ``$00$'' with the intermediate node being  a slave node.}
\item Each source symbol is  assigned to one node in $T_0$ and one node in $T_1$. \\The nodes to which they are assigned are either leaves  or master nodes.\\
{\em Symbols are not assigned to complete internal nodes or slave nodes.}
\item The root of $T_1$ is complete and its ``$0$'' child is a slave node.\\
The  root of $T_1$ has no  ``$00$'' grandchild.

\end{itemize}

Let $c_s(a), s\in \{0,1\}$ denote the codeword of $a\in \mathcal{X}$ encoded by $T_s$. 
The encoding procedure for a sequence  $x_1, x_2\ldots$ of source symbols works as follows.\\
0.  Set $s_1 = 0$ and $j=1.$\\
1. Encode $x_j$ as $c_{s_j}(x_j).$\\
2. If $c_{s_j}(x_j)$ is a leaf  in $T_{s_j}$, then set $s_{j+1} =0$ \\ \hspace*{.4in} else   set 
$s_{j+1} =1$   \quad \%\ this occurs when $c_{s_j}(x_j)$ is a master node  in $T_{s_j}$ \\
3. Set $j=j+1$ and Goto 1.

\medskip

Note that a symbol is encoded using $T_0$ if and only if its predecessor was encoded using  a leaf node and it is encoded 
using $T_1$ if and only if its predecessor was encoded using  a master node. 
The decoding procedure is  a straightforward reversal of the encoding procedure. Details are provided in 
\cite{original} and \cite{dp}.  The important observation is that identifying the end of  a codeword might first require reading an extra two bits past its ending, resulting in a two bit delay,  so decoding  is not instantaneous.

\begin{example} The trees in 
\Cref{fig:BBB} encode  $X= {\bf b}d {\bf b}c {\bf a}a$   as
$Y=y_1 \ldots y_{12}.$
$$
\begin{array}{|c|c|c|c|c|c|c|c|c|c|c|c|c|}
\hline
i =  & 1 & 2 & 3 & 4 & 5 & 6 & 7 & 8 & 9 & 10 & 11 & 12 \\\hline
y_i = &1  &0 &  1&  1&  0 & 0 & 1  & 0&  1 & 1 & 0 & 1\\ \hline
\end{array}
$$
Also write $X=x_1 x_2 x_3 x_4 x_5 x_6$.
To decode 
\begin{enumerate}
\item It  is known that $x_1$ was encoded using $T_0$. After reading $y_1 y_2= 10$, the decoder does not know whether $x_1=b$ or if $1 0$ is just the prefix of some longer coding word.   After reading $y_3= 1$ (a delay of 1 bit) it  knows for  certain that $x_1=b.$ 
\item Since $x_1$ was encoded using a master node, it is known that  $x_2$ was encoded using $T_1$. After reading $y_3 y_4= 1 1$ the decoder is at node $c$ in $T_1$ but does not know whether $x_2=c$ or not.  After reading $y_5=0$, it still does not know.  It is possible that $x_2=c$ and $y_5$ is the beginning of the encoding of $x_3.$  It is also still possible that $x_2=d.$  After reading $y_6=0$ it knows for certain that  $x_2 =d.$
\item Since $x_2$ was encoded using a leaf, $y_3$ was encoded using $T_0.$ 
\item Continuing, the process finds all of $X= {\bf b}d {\bf b}c {\bf a}a.$
\end{enumerate}
\end{example}

\begin{algorithm*}[h]
  \caption{Iterative algorithm to construct an optimal binary AIFV-$2$ code \cite{yamamoto2015almost,dp}}\label{DIV}
  \begin{algorithmic}[1]
  \State $m \leftarrow 0;$\  $C^{(0)}=2-\log_2(3)$
  \Repeat
    \State $m\leftarrow m+1$
    \State $T_0^{(m)}=\text{argmin}_{T_0}\{L(T_0)+C^{(m-1)}q_1(T_0)\}$
    \State $T_1^{(m)}=\text{argmin}_{T_1}\{L(T_1)-C^{(m-1)}q_0(T_1)\}$
    \State Update  $$C^{(m)}=\frac{L\left(T_1^{(m)}\right)-L\left(T_0^{(m)}\right)}{  q_1\left(T_0^{(m)}\right) + q_0\left(T_1^{(m)}\right)  }$$
  \Until $C^{(m)}=C^{(m-1)}$
  \State //  Set $C^* = C^{(m)}. $ Optimal binary AIFV-$2$ code is $T_0^{(m)}$, $T_1^{(m)}$
  \end{algorithmic}
  \label{alg: alg1}
\end{algorithm*}

Following \cite{original}, we can 
now  derive the average codeword length of a binary AIFV-$2$ code  defined by trees $T_0, T_1$.  The average codeword length $L(T_s)$ of $T_s$, $s \in \{0,1\},$  is 
$$L(T_s)=\sum_{i=1}^{n}|c_s(a_i)|p_i.$$

If the current symbol $x_j$ is encoded by a leaf (resp. a master node) of $T_{s_j}$, then the next symbol $x_{j+1}$ is encoded by $T_0$ (resp. $T_1$). This process can be modelled as  a two-state Markov chain with the state being the current encoding tree. 
Denote the transition probabilities for switching from code tree $T_s$ to $T_{s'}$ by $q_{s'}(T_s)$. Then, from the definition of the code trees and the encoding/decoding protocols:
$$q_0(T_s)=\sum_{a\in \mathcal{L}_{T_s}}P_X(a)
\quad \mbox{and} \quad 
q_1(T_s)=\sum_{a\in \mathcal{M}_{T_s}}P_X(a)$$
where $\mathcal{L}_{T_s}$ (resp. $\mathcal{M}_{T_s})$ denotes  the set of source symbols $a\in \mathcal{X}$ that are  assigned to a leaf node (resp. a master node) in $T_s$.

Given binary AIFV-$2$ code $T_0,T_1,$ 
as the number of symbols being encoded approaches infinity, the stationary probability of using  code tree $T_s$ can then be calculated to be 
\begin{equation}
\label{eq:2Dnondegen}
P(s\,|\, T_0, T_1)=\frac{ q_s(T_{\hat{s}})}{q_1(T_0)+q_0(T_1)}
\end{equation}
where
$\hat{s} = s -1.$

The average (asymptotically) codeword length (as the number of characters encoded goes to infinity)  
of a binary AIFV-$2$ code
is then
\begin{equation}
\label{eq:LAIFV def}
L_{AIFV}(T_0, T_1)=P(0\,|\,T_0, T_1)L(T_0)+P(1\,|\,T_0, T_1)L(T_1).
\end{equation}

\cite{original, yamamoto2015almost} showed that the  binary AIFV-$2$ code $T_0,T_1$ minimizing \Cref{eq:LAIFV def} 
 can be obtained by \Cref{alg: alg1}.  In Lines, 4 and 5,  ${\mathcal T}_s(n)$ for $s \in \{0,1\}$ is the (exponentially sized) set of all possible coding trees $T_s$ for $n$ characters.
 They implemented the minimization (over all $T_s \in {\mathcal T}_s(n)$)    as an  ILP.
 In a later paper, \cite{dp}, the authors replaced this ILP with a  $O(n^5)$ time and $O(n^3)$ space
 DP 
 that modified a top-down tree-building Dynamic  Program  from 
\cite{golin,chan2000dynamic}.

\cite{yamamoto2015almost,dp} proved algebraically that \Cref{alg: alg1}  would terminate after a finite number of steps and that the  resulting trees $T_0^{(m)},$  $T_1^{(m)}$ are an optimal Binary AIFV-$2$ code. They were unable, though, to provide any bounds on the number of steps needed for termination.


\section{Lower Envelopes and the Geometry of Binary AIFV-$2$ Codes}
\label{sec:LE and Geo}

In this section we describe a new way of viewing \Cref{alg: alg1},  providing an alternative interpretation of its  behavior.  More importantly, this view will also lead to the development of the two new  procedures for solving the same problem.

Implicit in the definition of  Binary AIFV-$2$ Codes is that, if $T_0$ or $T_1$ is part of an optimal coding tree pair, then every master node and leaf have associated source symbols and that a slave node does not have a slave node child (otherwise the code cost  can be improved).  This trivially  implies that the tree contains at most $3n$ nodes and therefore has height at most $3n$.  These loose bounds will be needed later.

\begin{figure*}[t]
 \center
\includegraphics[width=2.9in]{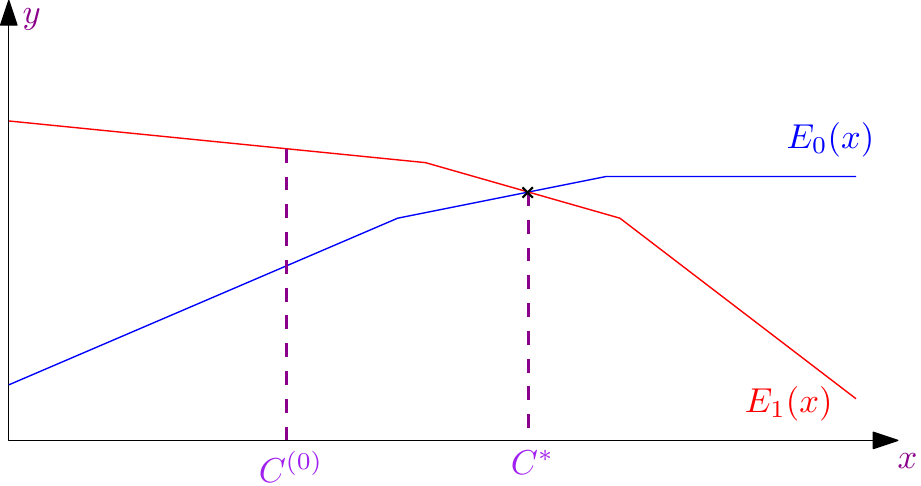}

\vspace*{.3in}

\includegraphics[width=2.9in]{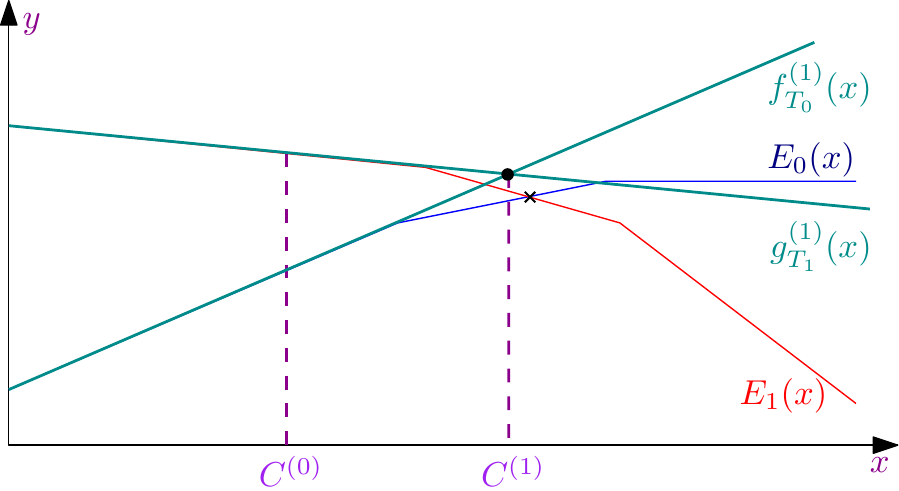}
  \caption{$E_0(x)$ is the lower envelope of the $f_{T_0}(x);$ $E_1(x)$  the lower envelope of the $g_{T_1}(x).$
  They intersect at a unique point  with $x^* = C^*.$ The bottom diagram shows the first step of \Cref{alg: alg1}.  It finds the supporting lines $f_{T_0}^{(1)}(x)$ and $g_{T_1}^{(1)}(x)$  at $x = C^{(0)}$ and then finds $C^{(1)},$ the $x$-coordinate of their unique intersection point.
  }
  \label{fig:CCC}
  \end{figure*}


\medskip

\begin{Note}
Let $s \in \{0,1\}$ and set $\hat s = 1 -s.$   Our algorithms and analyses do not use the actual trees but only the values of  $L(T_s)$ and   $q_{\hat s}(T_s) = 1 -q_{\hat s}(T_{\hat s}).$

We therefore say that $T_s, T'_s  \in {\cal T}_s(n)$ are {\em equivalent} if  $L(T_s) = L(T'_s)$ and  $q_{\hat s}(T_s)  = q_{\hat s}(T'_s).$
In particular, this permits stating that, for any $\alpha \in [0,1]$, and $\beta > 0$ there exists at most one tree $T_s \in {\cal T}_s(n)$ with  $q_{\hat s}(T_s) = \alpha$ such $L(T_s) = \beta.$
\end{Note}

\medskip
Let $s \in \{0,1\}$ and set $\hat s = 1 -s.$
Observe that for    
$$\forall T_s\in \mathcal{T}_s(n), \quad L(T_s) >  0 \quad\mbox{and} \quad  0 \le q_{\hat s}(T_s)\le 1.$$
 Also, there exist\footnote{While it is ``intuitive'' that a tree $T_1$ with $q_0(T_1)=0$ can not be part of an optimal Binary AIFV-$2$ code, the definition still permits its existence so we must include it in our base analysis. Lemma \ref{lem:nodegb} later will formally  validate this intuition.}
  trees  $T_s$ such that $q_{\hat s}(T_s)=0$ and also exist trees  $T_s$ such that $q_{\hat s}(T_s)=1.$ 
%
%
 In what follows, refer to \Cref {fig:CCC}.

\medskip

For every  fixed $T_0\in \mathcal{T}_0(n)$ define the linear function
$$f_{T_0}(x)=L(T_0)+q_1(T_0)x.$$

Now consider the lower envelope $E_0(x)$ of these lines.
More formally, 
$$\forall x\in \mathbb{R},  \quad \mbox{set}  \quad  E_0(x)=\min_{T_0 \in \mathcal{T}_0(n)}  f_{T_0}(x) = \min_{T_0 \in \mathcal{T}_0(n)} \{L(T_0)+q_1(T_0)x\}.$$
As the lower envelope of a finite set of lines, $E_0(x)$ is concave  so it is piecewise linear with decreasing slope. 
 Thus,  for small enough $x$, $E_0(x)$ has slope $1$ and for large enough $x$ it has slope $0.$

Similarly, for every  fixed $T_1\in \mathcal{T}_1(n)$, define the linear function
$$g_{T_1}(x)=L(T_1)-q_0(T_1)x$$
and consider the lower envelope $E_1$ of $g_{T_1}(x)$.
More formally,
$$\forall x\in \mathbb{R},  \quad \mbox{set}  \quad E_1(x)=\min_{T_1 \in \mathcal{T}_1(n)}  g_{T_1}(x)  =\min_{T_1 \in \mathcal{T}_1(n)} \{L(T_1)-q_0(T_1)x\}.$$
$E_1$ is also concave  so it is also piecewise linear with decreasing slope. 
Thus,
 for small enough $x$, $E_1(x)$ has slope $0$ and for large enough $x$ it has slope $-1.$
Hence,  
\begin{lem}
\label{lem:intsol}
$E_0(x)$ and  $E_1(x)$ intersect.  More specifically a real solution to 
$E_0(x) = E_1(x)$ exists and either 
\begin{itemize}
\item[$(S_1)$] the solution is a unique  point $x^*$ or 
\item[$(S_2)$] the solution  is an interval $[\ell,r]$ such that 
$\forall x^* \in [\ell,r],$
\begin{equation}
\label{eq:C0p}
\hspace*{-.2in} q_1(T_0(x^*)) = q_0(T_1(x^*)) = 0
\quad\mbox{and}\quad
L(T_0(x^*)) = L(T_1(x^*)).
\end{equation}
\end{itemize}
\end{lem}
Now set
\begin{eqnarray*}
T_0(x) &=& \text{argmin}_{T_0 \in \mathcal{T}_0(n)} \left\{ f_{T_0}(x)\right\},\\
T_1(x) &=& \text{argmin}_{T_1 \in \mathcal{T}_1(n)}  \left\{g_{T_1}(x)\right\}.
\end{eqnarray*}

In this notation, lines 4, 5 in \Cref{alg: alg1}  can be rewritten as
$$T_0^{(m)}=T_0\left( C^{(m-1)}\right)
\quad\mbox{and}\quad
T_1^{(m)}=T_1\left( C^{(m-1)}\right).$$

\Cref {alg: alg1}  can now be seen as implementing the following dynamic process on $E_0(x)$ and  $E_1(x).$
Line 1 starts by setting some value   $x_0 = C^{(0)}$ as an initial guess for $C^{*}$.
At the end of  step $m-1,$  the algorithm has $x_{m-1} = C^{(m-1)}.$ Then
\begin{itemize}
\item  In  lines 4, 5  it finds $T_0(x_{m-1})$,  $T_1(x_{m-1})$, the  indices of the respective  supporting lines of $E_0(x)$, $E_1(x)$ at $x=x_{m-1}.$
\item Line 6 (working through the algebra) then finds the unique intersection point of those two supporting lines
$ f_{T_0^{(m-1)}}(x)$   and  $ g_{T_1^{(m-1)}}(x)$ and sets $x_m =C^{(m)}$ to be the $x$-coordinate of that point.
\end{itemize}

The analysis in \cite{yamamoto2015almost,dp} implicitly assumes scenario $(S_1)$ in  Lemma \ref{lem:intsol}, i.e., that 
$E_0(x) = E_1(x)$  has a unique solution  $x^*.$
Line 7 can then only be satisfied if $C^{(m-1)}= x^*$ and their analysis
of  \Cref{alg: alg1}  can then be recast as saying that that the algorithm converges in a finite number of steps to $x^*,$ which gives the optimal solution $T_0(x^*),$  $T_1(x^*).$

Using our notation, their proof can be restated as:
\begin{lem} \cite{yamamoto2015almost,dp}
\label{lem:sol}
If  $x^*$  is a unique real solution to $E_0(x^*) = E_1(x^*)$, then $T_0(x^*),\, T_1(x^*)$ is an optimal Binary AIFV-$2$ code. Furthermore, \Cref{alg: alg1} will find that code.
\end{lem}

The difficulty in analyzing   the running time of the algorithm is that, a-priori, it is possible that all the (exponential number of) trees  $T_s\in \mathcal{T}_s(n)$
contribute   a supporting line of the lower envelope $E_s(x)$.
 The iterative process might then need to examine all of these supporting lines, which would be very slow. 
In the next section  we describe two new algorithms  that sidestep this issue to provide explicit bounds on the number of iterations needed.  


Before concluding we quickly discuss two assumptions made by \cite{yamamoto2015almost,dp} 
\begin{itemize}
\item [($A_1$)]  As previously noted, they assume $(S_1)$ in  Lemma \ref{lem:intsol}, i.e., that 
$E_0(x) = E_1(x)$  has a unique solution  $x^*.$ \\ If  $(S_2)$ occurred, then   Line 6 in \Cref{alg: alg1} would not be feasible and Lemma \ref{lem:sol} would not  immediately follow. (Their proof could be modified so that this assumption  is not necessary, though).
\item[($A_2$)] Starting with $C^{(0)} =  2 - \log_2(3)$, they assume that, in
\Cref {alg: alg1}, for all $m,$   $C^{(m)} \in [0,1]$.  \\  $C^{(m)} \in [0,1]$ was  implicitly required for the correctness of their $O(n^5)$  Dynamic Program for implementing Lines 4 and 5 of \Cref{alg: alg1}. Without this assumption, the DP would need to be replaced by the less efficient Integer Linear Program.   
\end{itemize}
Our new algorithms will need similar assumptions, but  replace ($A_2$) with the weaker condition
\begin{itemize}
\item[($A'_2$)]
 $x^* \in  [0,1]$. 
 \end{itemize}
  We  will need both of these assumptions even more strongly for our algorithms and therefore provide proofs of the validity of $A_1$ and $A'_2$ in the next subsection.

  \subsection{Proving the Assumptions ($A_1$)  and   ($A'_2$).}
\label{sec:cinunitint}

Before preceding we note that various places in this section discuss  ``master nodes with  no codeword descendants''.   Since  a master node must have an edge  descending from it, this might, at first glance, seem impossible.   The formal definition, though, only requires the existence of the {\em edge} but not that the edge leads to codewords. 

It is also not a-priori obvious that removing that edge (and everything below it) would keep the code cost unchanged.
Tree edges are not only  {\em links} to  codewords {\em below} them.  They may  also {\em signal} that a codeword {\em above } them is a master node,  forcing a  transition  to a new codetree,  and it is this transition  that could  be making the code optimal.  Removing unused edges might thus  {\em increase} the cost of the code. 
The reality is that this will not be an issue but this must be proven and can not be assumed.

\begin{figure}[t]
 \center
  \includegraphics[width=2.3in]{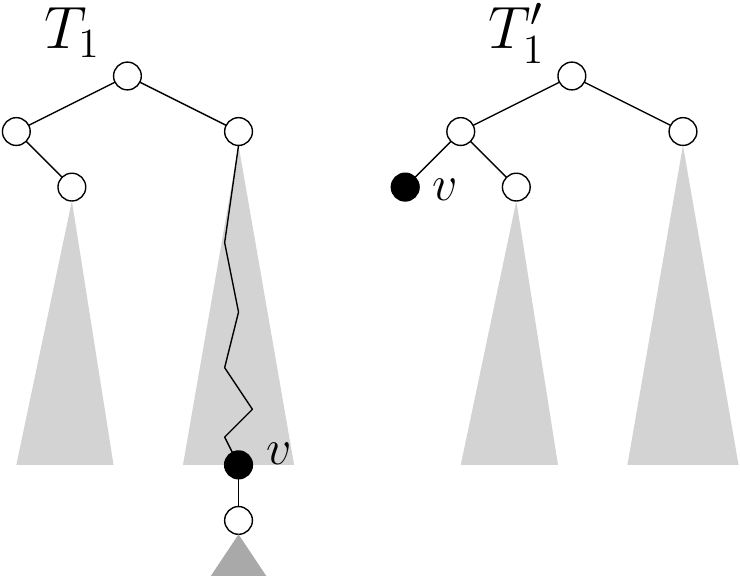}
 \caption{The construction in  the proof of  Lemma \ref{lem:t1p transform}.}
 \label{AAAX}
\end{figure}

We first need the following simple observation.\footnote{The authors thank Haoren Zhu for pointing this out, permitting the simplification of the later lemmas.}
\begin{lem}
\label{lem:t1p transform}
Assume  $n > 3.$ Let $T_1\in \mathcal{T}_1(n)$ and $a_i$ be the longest codeword in $T_1$ with $v$ the associated node in the tree. Remove $v$ and any descendants from $T_1$.  Make $a_i$ the left-left child of the root of $T_1$. Call the resulting tree $T'_1$.  Then
\begin{itemize}
\item  $T'_1\in \mathcal{T}_0(n)$
\item $L(T'_1) < L(T_1).$
\end{itemize}
\end{lem} 
\begin{proof} See \Cref {AAAX}.
Note that, formally, $v$ might be  a master node (and thus with an edge descending below it)  but no codewords below it.  This is why $v$ {\em and its descendants} are all removed.

Since $T_1$'s left child did not originally have a left child, $a_i$ can be moved there to become the leftmost grandchild of the root. 
$T'_1\in \mathcal{T}_0(n)$ by definition.   $n >3$ ensures that $depth(v)$ in $T_1$ was at least 3. $ L(T'_1) < L(T_1)$  follows.
\end{proof}

\begin{figure}[t]
 \center
  \includegraphics[width=2.3in]{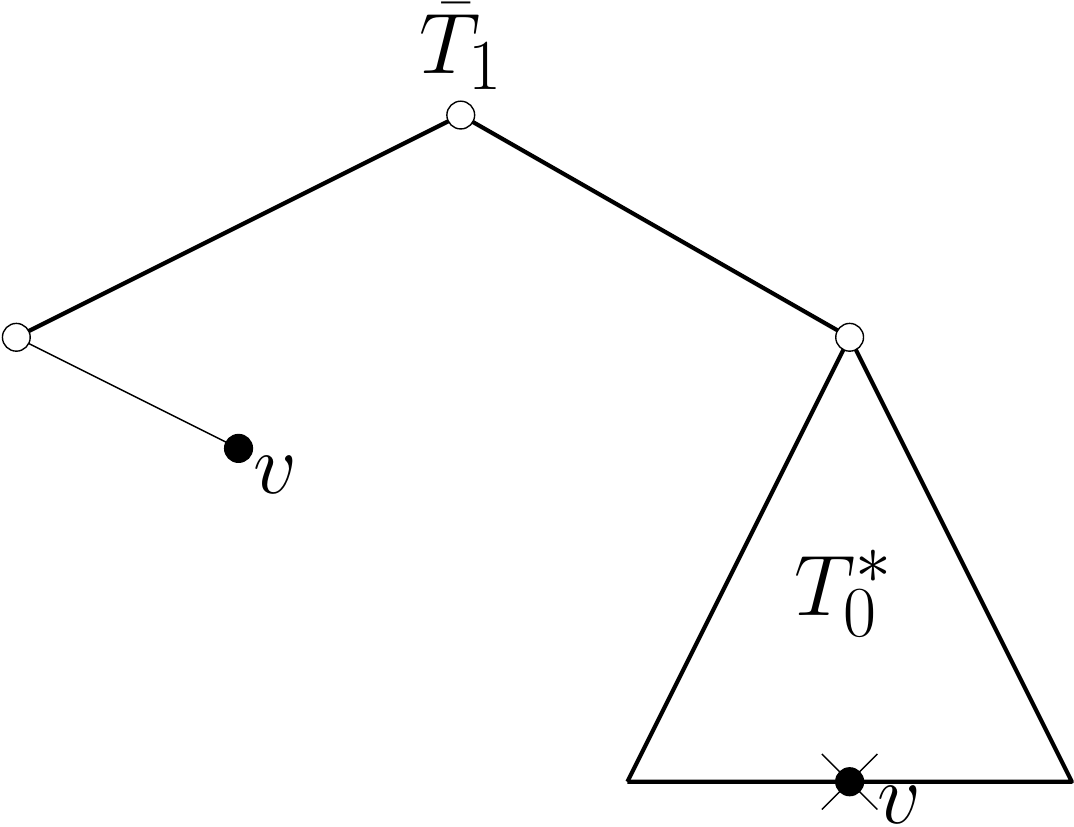}
 \caption{In the proof of  Lemma \ref{lem:inint},  $\bar T_1$ is constructed by taking a leaf $v$ of $T_0^*$ and making it $\bar T_1$'s  $01$ child  and then appending the remaining $T_0^*$ as
  $\bar T_1$'s   child.}
 \label{AAA}
\end{figure}

\begin{lem}  There exists $x^* \in (0,1]$ such that $E_0(x^*) = E_1(x^*).$
\label{lem:inint}
\end{lem}
\begin{proof}
Recall that  $E_0(x)$ is concave and nondecreasing, and $E_1(x)$ is concave and nonincreasing.  It therefore suffices to show 
 that $E_0(0) < E_1(0)$ and that $E_0(1)\geq E_1(1)$. Then, by continuity $E_0(x^*)= E_1(x^*)$ for some  $x^* \in (0,1].$


Consider tree $T_1= T_1(0).$  From Lemma \ref{lem:t1p transform},
$$
E_0(0) = \min_{T_0\in \mathcal{T}_0(n)}\{L(T_0)\}\\
             \le  L(T'_1) \\
             < L(T_1)
	    = E_1(0),
$$
proving $E_0(0) < E_1(0)$.

To  prove  $E_0(1) \geq E_1(1)$, 
fix $x=1$ and let $T_0^*= T_0(1)$ and  $T_1^*=T_1(1),$ 
i.e. 
$$
  E_0(1) = L(T^*_0)+q_1(T^*_0)
\quad\mbox{and}\quad
  E_1(1) = L(T^*_1)-q_0(T^*_1).
$$
and
\begin{equation}
\label{eq:uniqueT0}
\forall T'\in \mathcal{T}_0(n),\quad  L(T_0^*)+q_1(T_0^*) \leq L(T') + q_1(T'),
\end{equation}
\begin{equation}
\label{eq:uniqueT1}
\forall T'\in \mathcal{T}_1(n),\quad  L(T_1^*)-q_0(T_1^*) \leq L(T') - q_0(T').
\end{equation}

Also note that $T^*_0$ must contain at least one leaf.  If not,   just changing the deepest master node into a leaf would create a new tree $T' \in  \mathcal{T}_0(n),$  with
$L(T') = L(T^*_0)$ but $q_1(T') <  q_1( T^*_0)$, contradicting \Cref{eq:uniqueT0}.

Now, we construct   (\Cref{AAA})  tree $\bar T_1$  by (i) removing  a leaf  from $T^*_0$ and making it a leaf on the $01$ path from the root of $\bar T_1$ and (ii) appending the remaining part of $T^*_0$ as the right child of the root of $\bar T_1.$ By this construction, $\bar T_1 \in  \mathcal{T}_1(n)$,
 $$L(\bar T_1)\leq L(T_0^*)+\sum_{i=1}^n p_i = L(T_0^*)+1$$
and 
$$q_0(\bar T_1)\geq q_0(T_0^*).$$
Thus, from \Cref{eq:uniqueT1},
\begin{eqnarray*}
L(T_1^*) - q_0(T_1^*) &\leq&  L(\bar T_1) - q_0(\bar T_1)\\
 &\leq& L(T_0^*)+1 - q_0(T_0^*), 
\end{eqnarray*}
so
$$L(T_0^*)-L(T_1^*) \geq q_0(T_0^*)-q_0(T_1^*)-1.$$
Now
\begin{eqnarray*}
E_0(1)-E_1(1)&=& L(T_0^*)+q_1(T_0^*)-L(T_1^*)+q_0(T_1^*) \\
&=& \left( L(T_0^*)-L(T_1^*) \right) + q_1(T_0^*)+q_0(T_1^*)\\
&\ge & \left( q_0(T_0^*)-q_0(T_1^*)-1 \right) + q_1(T_0^*)+q_0(T_1^*)\\
&=& 1 -1 = 0.
\end{eqnarray*}
$E_0(1)\geq E_1(1)$ follows.
\end{proof}

\begin{lem}\ 
\label{lem:nodegb}
\begin{enumerate}
\item[(a)] The solution to  $E_0(x^*) = E_1(x^*)$ is unique,  satisfies  $x^* \in (0,1]$ and  
\begin{equation}
\label{eq:x*horiz}
q_0(T_1(x^*)) >0.
\end{equation}
Consequentially, Scenario $(S_2)$ in Lemma \ref {lem:intsol} does not occur.
\item[(b)] Let $T_0,T_1$ be an optimal Binary AIFV-$2$ code.  If $q_1(T_0) = 0$, then  
$$\forall \bar T_1 \in \mathcal{T}_1(n),\quad L_{AIFV}(T_0, \bar T_1)=L(T_0)$$
so $\forall \bar T_1 \in \mathcal{T}_1(n),$
$T_0, \bar T_1$ is  an optimal Binary AIFV-$2$ code.
\end{enumerate}
  \end{lem}
\begin{proof}
(a) 
From   Lemma \ref{lem:inint} we may assume that a solution $x^* \in (0,1]$ exists.
Let  $T_0= T_0(x^*)$,  $T_1= T_1(x^*)$ and assume that \Cref{eq:x*horiz} is incorrect, i.e. $q_0(T_1) =0$.

Let $a_i$ be the longest codeword in $T_1$ with $v$ the associated node in the tree. Because  $q_0(T_1) =0,$  $v$ is a master node.  Remove the descendants of  $v$,  so $v$ is now  a leaf.  Call the resulting tree $\bar T_1.$  By construction,  $\bar T_1 \in {\mathcal T}_1(n)$,  $L(T_1) = L(\bar T_1)$ and $q_0(\bar T_1) = p_i >0.$
But then
\begin{eqnarray*}
E_1(x^*) &=& \min_{T_1\in \mathcal{T}_1(n)}\{g_{T_1}(x^*)\}\\
              &\le& g_{\bar T_1} (x^*)\\
              &=& L\left(\bar T_1\right) - q_0\left(\bar T_1\right)x^*\\
              &= & L(T_1) - q_0\left(\bar T_1\right)x^*\\
              &<&  L(T_1) \\
              &=& E_1(x^*),
\end{eqnarray*}
leading to a contradiction.

But if  $(S_2)$ occured,  then, from  \Cref{eq:C0p}, 
for {\em all } solutions $x^*$ to to $E_0(x^*) = E_1(x^*)$, $q_0(T_1(x^*))=0,$  contradicting the above.  Thus
$(S_2)$ can not occur and $x^*$ is unique.


(b)   Let $T_0, \bar T_1 $ be any Binary AIFV-$2$ code.  Since $q_1(T_0)=0,$   by the coding procedure, $\bar T_1$ is never used, so
$$P(0\,|\,T_0, \bar T_1) =1
\quad\mbox{and} \quad 
P(1\,|\,T_0, \bar T_1)=0.$$
Thus, from \Cref{eq:LAIFV def}
$$L_{AIFV}(T_0, \bar T_1)=L(T_0).$$
In particular,    
$$L_{AIFV}(T_0, \bar T_1)   =L_{AIFV}(T_0, T_1).$$

\end{proof}

Lemma  \ref {lem:nodegb}(a) implies  that assumption ($A_1)$ is always true, so the condition of  Lemma \ref{lem:sol} always applies and thus \Cref {alg: alg1} of \cite{yamamoto2015almost,dp} always  returns the correct answer.

Lemma \ref{lem:inint}  is slightly weaker than assumption ($A_2$) since it only states that $x^* \in [0,1]$ but does not immediately imply that $C^{(m)} \in [0,1]$ for all $m.$  While it is possible to strengthen the proof to validate ($A_2$), we do not do so since our algorithms only needs the weaker assumption ($A'_2$) to be correct.

Finally, we note that Lemma \ref{lem:nodegb} does  permit an optimal Binary AIFV-$2$ code to satisfy $q_1(T_0) = 0$.
In this case, $T_0$ is  the optimal tree in which all codewords are at leaves, i.e., the optimal Huffman code and AIFV coding reduces to Huffman coding.  This could occur, for example,  if the source distribution $X$ is dyadic so $L(T_0) = H(X),$ the entropy of $X$.  It is  therefore understandable that in this case,   AIFV codes could  not improve upon Huffman coding.  We single out  this special case, because our Ellipsoid-based algorithm will need to perform extra work to handle it.

\section{$O(n^5 b)$  time  algorithms for constructing binary AIFV-$2$ codes}
\label{sec: new algs}
\subsection{An $O(n^5 b)$  time Binary Search  Algorithm}
\label{subsec:bsearch}

\begin{algorithm}[t]
  \caption{An $O(n^5 b)$ Binary Search algorithm for constructing binary AIFV-$2$ trees}\label{DIV}
 \label{alg: alg2}
  \begin{algorithmic}[1]
  \State $l,r \leftarrow 0,1$; $\epsilon_0 = 2^{-2(b+1)}.$
  \Repeat
 \State $\textrm{mid}\leftarrow \frac{l+r}{2}$;
 \State $ e_0= E_0(\textrm{mid}) = f_{T_0(\textrm{mid})}(\textrm{mid})$; 
 \State $ e_1 = E_1(\textrm{mid})= g_{T_1(\textrm{mid})}(\textrm{mid})$
    \If {$e_0 <e_1$}
    	\State $l =\textrm{mid}$
    \Else
    	\State $r=\textrm{mid}$
    \EndIf
\Until $r-l = \epsilon_0$
 \\
 \State//  Now find  $x^* \in [l,r],$
  \State Construct   $T_0(l)$, $T_1(l)$,    $T_0(r)$ $T_1(r)$ and their corresponding lines.
  \State In $O(1)$ further time (as described at end of  \Cref{subsec:bsearch}) find $x^* \in [l,r],$ the 
   intersection point of $E_0(x)$ and $E_1(x),$
    \State  Optimal AIFV code is $T_0(x^*)$ $T_1(x^*)$
  \end{algorithmic}
\end{algorithm}

The geometric interpretation in Section \ref{sec:LE and Geo}  suggests  using a simple  \textbf{binary search}  to replace the iterative process of \Cref {alg: alg1}. This  is written in  \Cref {alg: alg2} and explained below.
\medskip

%
%

The $O(n^5)$ time DPs from \cite{dp} for evaluating lines 4, 5 in \Cref {alg: alg1} are actually explicitly constructing the trees  $T_0(x)$ and $T_1(x)$  for any fixed $x \in [0,1].$   We will use those same construction DPs in \Cref {alg: alg2}. Note that after constructing  trees $T_0(x)$ and $T_1(x)$, 
values $E_0(x)$ and $E_1(x)$  can be evaluated  in an additional  $O(n)$ time and whether $E_0(x) \le E_1(x)$ or not can then be checked in  another $O(1)$ time.

Lemma \ref{lem:intsol}  implies that,
for any $l < r$, 
$$ E_0(l) \le E_1(l)  \mbox{ and }  E_0(r) \ge E_1(r)
\quad  \mbox{if and only if} \quad  \exists x^* \in [l,r]$$
where $x^*$ the unique  intersection point of $E_0(x),$ and $E_1(x).$

From Lemma \ref{lem:inint},  $x^* \in [0,1]$. We therefore  use a binary search halving algorithm to maintain an interval  $[l,r]\subseteq [0,1]$ containing $x^*.$  
 Start by setting $[l,r] = [0,1]$. 

 At each step, set  $\textrm{mid}= (l+r)/2$ and check  
if  
 $E_0(\textrm{mid}) \le E_1(\textrm{mid})$.  If yes, set $[l,r] = [\textrm{mid},r]$.  If no, set 
 $[l,r] = [l,\textrm{mid}]$.  From the discussion above, this maintains  $ x^* \in [l,r].$

Since \Cref {alg: alg2} always maintains $[l,r] \in [0,1]$, all calls to the dynamic programs in  lines 4, 5, 13 and  15 of  \Cref {alg: alg2} are correctly finding $T_0(x),\, T_1(x)$ in $O(n^5)$ time.

Since the procedure halves $r-l$ at each step, 
after $\log_2 \frac{1}{\epsilon_0}= 2(b+1) $ steps,  $r-l = \epsilon_0$ and the algorithm proceeds to Line 13.

\medskip

Finally,  recall that $E_0(x)$ and $E_1(x)$ are piecewise linear functions.  Their {\em critical points} will be the values of $x$ at which they change slope.  

\begin{lem}
\label{lem:spacing}
 Let $x_1,x_2$ be two critical points of $E_0(x)$ (resp. $E_1(x)$). Then $$|x_1 - x_2| \ge 2^{-2(b+1)}.$$
\end{lem}
\begin{proof}
We prove the lemma for $E_0(x)$.  The proof for  $E_1(x)$ is almost exactly the same.

Let $f_{T_0}(x)$ and  $f_{T'_0}(x)$ be two supporting lines of $E_0(x)$ that meet at a critical point $ x'.$  
Then $$x' = \frac {L(T'_0)- L(T_0)} {q_1(T_0)+q_1(T'_0)}.$$ 

Since every $p_i$ can be expressed using $b$ bits, every $p_i$ can be written as  $c 2^{-b}$ for some integer $1 \le c  \leq 2^b.$
This implies  
\begin{eqnarray*}
q_1(T_0) &=& c_1 2^{-b},\\
q_1(T'_0) &=&  c_2 2^{-b},\\
L(T_0) &=& c_3 2^{-b},\\
L(T'_0) &= &c_4 2^{-b}
\end{eqnarray*}
for some integers $c_i$  satisfying  $0  \le c_1,c_2 \leq 2^{b}$, 
and  $0 < c_3,c_4 \le  3n 2^{b}$.

This further implies that $x' = \frac {K_1} {K_2}$ where $K_1, K_2$ are integers and $0 \le  K_2 \leq 2^{b+1}.$  

This is true for every critical point $x'$.  Thus,  if $x_1,x_2$  are two critical points of $E_0(x)$, then $|x_1-x_2| \geq \left( \frac 1 {2^{b+1}}\right)^2 = 2^{-2(b+1)}.$
\end{proof}

As noted, at the start of  Line  13,  $r-l = \epsilon_0.$
From  Lemma \ref{lem:spacing},  $[l,r]$  therefore contains at most one critical point each for $E_0(x)$ and $E_1(x).$  This means that
$$
\forall x \in [l,r],\quad
\begin{array}{rcl}
E_0(x) &=& \min\left(f_{T_0(l)}(x),f_{T_0(r)}(x)\right)\\
E_1(x) &=& \min\left(g_{T_1(l)}(x),g_{T_1(r)}(x)\right)
\end{array}
$$

Note that it is possible that $f_{T_0(l)}(x) =f_{T_0(r)}(x)$ and, likewise, that $g_{T_1(l)}(x)= g_{T_1(r)}(x)$.

By construction, $[l,r]$   also contains the unique   intersection point $x^*$  of  $E_0(x)$ and $E_1(x)$.

Since,  in $[l,r],$ $E_0(x)$ and $E_1(x)$ are each composed of either one or two known line segments we can, in
$O(1)$ time,  calculate  $x^*$. 

In another $O(n^5)$ time we calculate $T_0(x^*),\, T_1(x^*)$ which, from Lemma \ref{lem:sol},    are an optimal binary AIFV code.

\medskip

To recap,  in $O(n^5 b)$ time the binary halving  algorithm found $[l,r]$ containing $x^*$ and in another $O(n^5)$ time found the optimal AIFV code $T_0(x^*)$ and $T_1(x^*)$
  Thus, the entire algorithm requires  $O(n^5 b)$ time.

Note that this is very different than  \Cref {alg: alg1} which also used $O(n^5)$ time per step but was only guaranteed  to converge in a finite but unbounded  number of steps.

\subsection{ An $O(n^5 b)$ time Ellipsoid Algorithm}
\label{subsec:ellipsoid}

\begin{figure*}[t]
 \center
\includegraphics[width=4.5in]{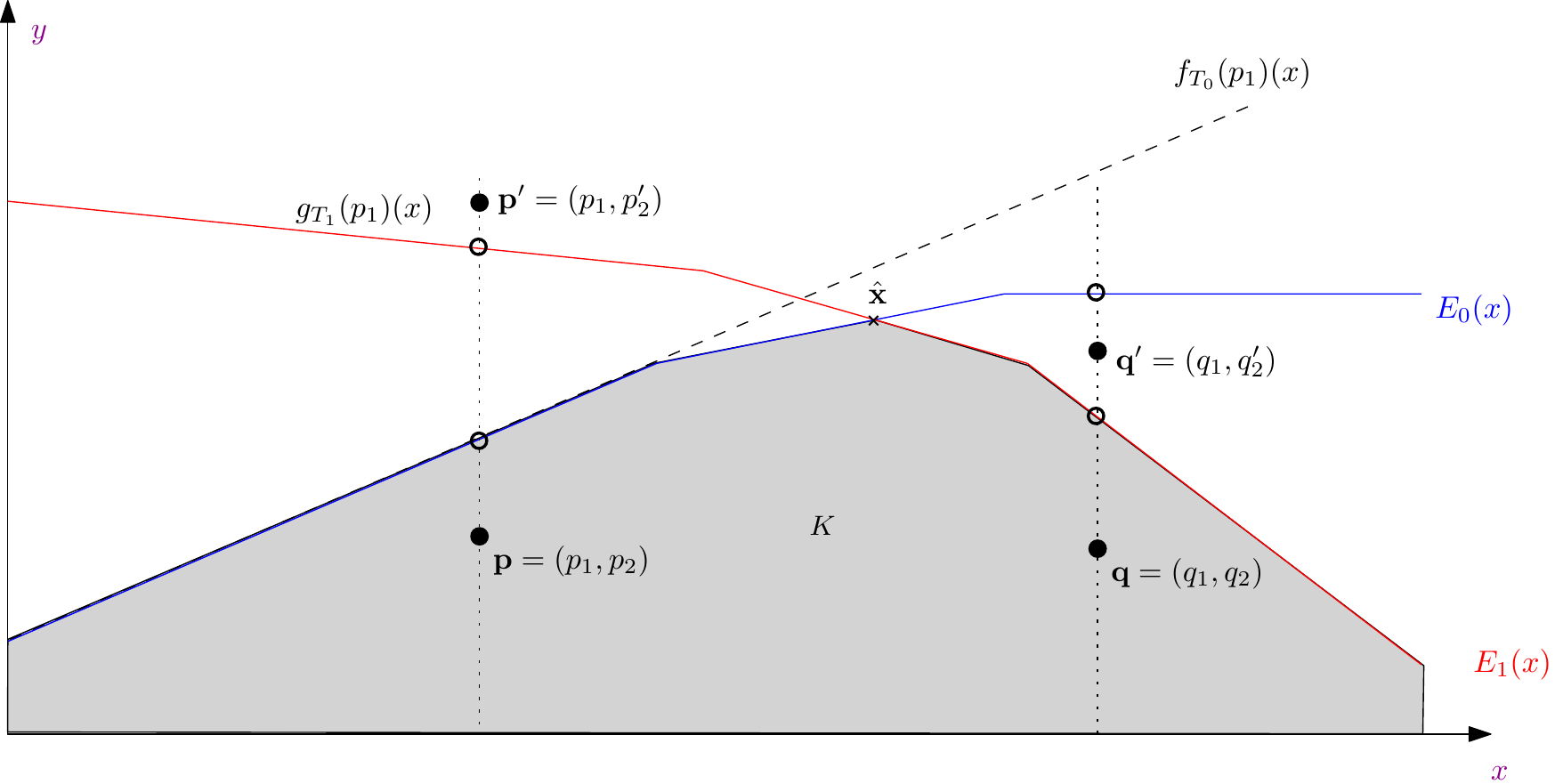}
  \caption{
 $ K = \bigl\{(x_1,x_2) \,:\,   0 \le x_1 \le 1 \ \mbox{and}  \  0 \le x_2 \le \min\left( E_0(x_1),\, E_1(x_1)\right) \bigr\}$ is convex. $\hat \bfx= (\hat x_1, \hat x_2)$, the highest point in $K,$ is the unique intersection of
 $E_0(x)$ and $E_1(x)$ so $\hat x_2 = E_0(x^*).$  
 Points $\bfp,\bfq \in K$, so a separation oracle would report that they are inside $K.$  Points $\bfp', \bfq'\not\in K$. $\bfp'$ is separated from  from $K$ by either one of  the lines  $f_{T_0(p'_1)}(x)$ or $g_{T_1(p'_1)}(x).$
 $\bfq'$ is separated from $K$ by the line $g_{T_1(q'_1)}(x)$ but not by the line $f_{T_0(q'_1)}(x)$.
%
%
  }
  \label{fig:DDD}
  \end{figure*}
  
We now introduce 
another $O(n^5 b)$ time algorithm  for solving the problem. While it is more complicated and no faster  than  \Cref{alg: alg2},  it is introduced because it provides some guidance as to how one might be able to solve the problem for  $m > 2.$

Let $K\in \mathbb{R}^m$ be a closed convex set. 
Let $c \in \mathbb{R}^m$.  The convex maximization problem is to find $\hat \bfx \in K$ such that
\begin{equation}
\label{eq:Kmax}
c^T \hat \bfx = \max \left\{c^T \bfx\,:\, \bfx\in K \right\}.
\end{equation}

See Figure \ref{fig:DDD}. In the binary AIFV-$2$ problem, $m=2$. Set 
\begin{eqnarray*}
M(x) &=&  \min\left( E_0(x),\, E_1(x)\right)\\
K &=& \bigl\{(x_1,x_2) \,:\,   0 \le x_1 \le 1 \ \mbox{and}  \  0 \le x_2 \le M(x_1) \bigr\}\\
c&=&(0,1).
\end{eqnarray*}

By the concavity of  $E_0(x)$ and $E_1(x)$,  $M(x)$ is concave so $K$ is convex.  

Now let $x^*$ be the unique  solution of $E_0(x) = E_1(x)$ and set $\bfx^* = \left(x^*,E_0(x^*)\right).$

Since $E_0(x)$ is a non-decreasing function, 
$$\forall x < x^*,\quad   M(x) \le E_0(x)  \le E_0(x^*).$$
Similarly, since  $E_1(x)$ is a non-increasing function, 
$$\forall x > x^*,\quad   M(x) \le E_1(x)  \le E_1(x^*).$$

Thus $\hat \bfx =  \bfx^*$  achieves the maximum value in  \Cref {eq:Kmax} with $c=(0,1).$  We can actually say something stronger

\begin{figure*}[t!]
 \center
\includegraphics[width=5in]{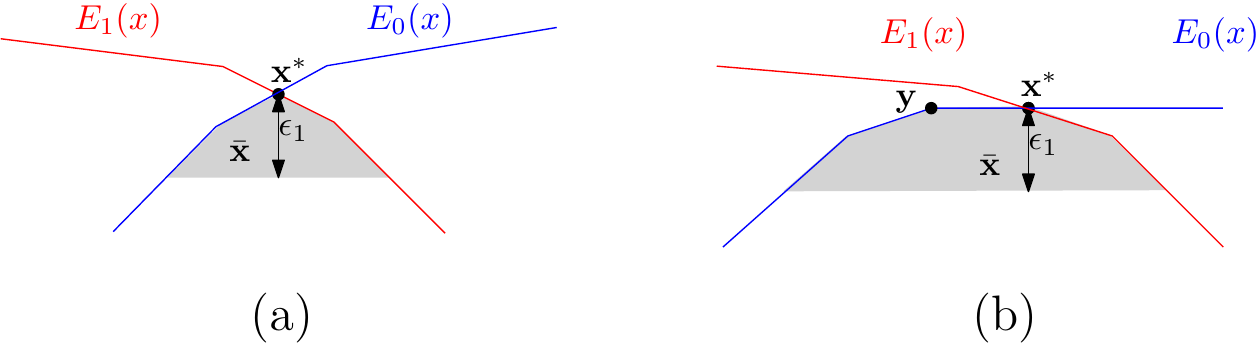}
  \caption{ Illustration of the two scenarios that can occur in Lemma \ref{lem:22cases} and \Cref{alg: ellipse2}. Note that these correspond to the similar areas in \Cref {fig:2Dgeo}. The shaded area in each diagram is the set of points $\bar x \in K$ such that 
  $E_0(x_1^*)  - \bar x_2 \le \epsilon_1$ and so might be returned by Line 2 of \Cref{alg: ellipse2}.  In scenario  (a), Line 12 will  find $x'= x_1^*.$
   In  scenario (b), Line 12  might not find $x^*_1$ but will return some $x' \in [y_1,x^*_1]$.
   In both cases   $T_0(x'),$  $T_1(x')$ is an optimal Binary AIFV-$2$ code.  In  (b),  $x'$ depends upon the location of  $\bar \bfx$; there are two possible trees  $T_1(x')$ that could be returned.  
Since  $q_1(T_0(x^*_1))=0,$ the returned tree $T_1(x')$ is never used though;  $T_0(x^*_1)$ is the optimal Huffman code which, in this case, is also the optimal AIFV-$2$ code.}
  \label{fig:4cases2D}
  \end{figure*}

\begin{lem}
\label{lem:22cases}
Let $\hat \bfx= (\hat x_1, \hat x_2)\in K$ be any point that  achieves the maximum value in  \Cref {eq:Kmax} with $c=(0,1).$  Then 
$T_0(\hat x_1),$ $T_1(\hat x_1)$ is an optimal Binary AIFV-$2$ code.
\end{lem}
\begin{proof}
See \Cref{fig:4cases2D}.

From  Lemma \ref{lem:sol},  $T_0(x^*)$,  $T_1(x^*)$ is an optimal Binary AIFV-$2$ code.  From  Lemma \ref{lem:nodegb},    $x^*$ is unique and $q_0(T_1(x^*)) >0.$   There are then only two possibilities:

\begin{itemize}
\item[\bf (a)]    $q_1(T_0(x^*)) >0$:
In this case,  $x^*$ is the unique solution to $\argmax_{x \in \Re} M(x)$. Thus, $\hat \bfx = \bfx^*$  and 
correctness immediately  follows from Lemma \ref{lem:sol}. 

\item[\bf (b)]   $q_1(T_0(x^*)) =0:$

Let $\bfy=(y_1,y_2)$ be the leftmost  point in $K$ such that $ y_2 = M(x^*).$  By the convexity of $K,$
 $\hat x_1 \in [y_1,x_1^*],$  $T_0(\hat x) = T_0(x^*)$ and  
$$M(\hat x) = E_0(\hat x) = E_0(x^*) = \max_{x \in \Re} M(x).$$		
 Lemma \ref{lem:nodegb} then implies that $T_0(\hat x), T_1(\hat x)$  is   an optimal Binary AIFV-$2$ code.
 
 Note that in this case $T_0(\hat x) = T_0(x^*)$ is the optimum Huffman code (and $T_1(\hat x)$ is never actually used).
\end{itemize}
\end{proof}

 A {\em separation oracle} for $K$ is  a procedure that, for any
$\bfx \in  \mathbb{R}^m$, either reports that $\bfx \in K$ or, if $\bfx \not \in K$, returns  a hyperplane that separates $\bfx$ from $K.$   That is, it returns $a \in \mathbb{R}^m$ such that  $\forall \bfz \in K,$ $a^T \bfx > a^T  \bfz.$

A famous result due to Gr{\"o}tschel,  Lov{\'a}sz and Schrijver \cite{ellipsoid} states that (even if $K$ is a polytope defined by an exponential number of constraints) an approximate optimal solution to \Cref{eq:Kmax}  can be solved via the ``ellipsoid method'' with only a ``logarithmic'' number of calls to the separation oracle. 
More explicitly, their result is
\begin{thm} \cite{ellipsoid}
\label{thm:ellipsoid}
Let $K\in \mathbb{R}^m$ be a closed convex set and $c\in \mathbb{Q}^m$. Assume   a known separation oracle for $K$. Also assume known positive values $R$ and $\epsilon$ such that $K\subset B(0,R)$ and $Vol(K)>\epsilon$. 
Then, $\bar x \in K$ satisfying 
$$c^T \bar x \geq \max \{c^Tx \,:\, x\in K \}-\epsilon_1 ||c||$$
can be found 
using at most  
$  3m(\log\frac{1}{\epsilon}+2m\log(2R)+m\log\frac{1}{\epsilon_1})$  separation oracle calls.
\end{thm}

\par\noindent{Note the following observations:}
\begin{itemize}
\item[(a)]  $\forall T_s\in \mathcal{T}_s(n)$,  the height of $T_s$ is $\le 3 n$. Thus  $\forall\, x \in[0,1],$ 
$f_{T_0}(x) \le  3n$
and  $g_{T_1}(x) \le 3n$. 
This implies
$$  \forall x \in[0,1],\quad
\min\left( E_0(x),\, E_1(x)\right) \le 3n,
$$ 
so $K\subset B(0,4n)$.

\medskip

\item[(b)] For all  $x \in[0,1],$  $f_{T_0(x)}(x) \ge L(T_0(x)) \ge 1$ so $E_0(x) \ge 1$.

Also,  $L(T_1(x)) \ge 1 \ge q_0(T_1(x))$ so $E_1(x) = g_{T_1(x)}(x) \ge 0$. Thus,  $M(x) \ge 0.$  Furthermore
$M(x^*) =E_0(x^*) \ge 1$.

From the convexity of $K$, the triangle with vertices $(0,0)$,  $(1,0),$  $(\bfx^*,1)$ is contained in $K,$ so
 $Vol(K) \ge 1/2.$

\medskip

\item[(c)] Separation Oracle: The DPs from  \cite{dp}  provide a simple $O(n^5)$ time separation oracle for  $\bfx=(x_1,x_2)$,  from the following observations:

\begin{itemize}
\item [(i)] If  $x_1 < 0,$  $x_1 > 1$ or $x_2 < 0$ report that $\bfx=(x_1,x_2) \not\in K$.  

The separating line is respectively  $x=0$  ($a =(-1,0)$),  \\ $x=1$  ($a =(1,0)$) or $y=0$ ($a =(0,-1)$).

\item[(ii)] Otherwise, $0 \le x_1 \le 1$  and  $ 0 \le x_2.$ 
Use the DPs to find $T_0(x_1)$ and $T_1(x_1)$  and thus
 $E_0(x_1),$   and $E_1(x_1).$ This is guaranteed to give a correct solution (in $O(n^5)$ time)  because $0 \le x_1 \le 1$ satisfies the requirements of the DPs in  \cite{dp} .
 
 \end{itemize}

 \begin{itemize}
 \item 
 If  $0 \le x_2 \le \min\left( E_0(x_1),\, E_1(x_1)\right)$ then report that $\bfx \in K.$
 \item  If not, report that $\bfx=(x_1,x_2) \not\in K$. \\ At least one of $x_2 > f_{T_0}(x_1)$ or  $x_2 > g_{T_1}(x_1)$ must hold.
\begin{itemize}
     \item If  $x_2 > f_{T_0}(x_1)$  \quad $\Rightarrow$  \quad line   $f_{T_0}(z)$ separates $\bfx$ from $K$. \hspace*{.07in} \\ $\bigl(a=(1,-q_1(T_0))\bigr)$
    \item If $x_2 > g_{T_1}(x_1)$  \quad $\Rightarrow$  \quad   line 
$g_{T_1}(z)$ separates $\bfx$ from $K$. \hspace*{.07in}  \\ $\bigl(a=(1,q_0(T_1))\bigr)$
 \end{itemize}
\end{itemize}
\end{itemize}

Thus, plugging into Theorem \ref{thm:ellipsoid}   with  $m=2$,  $R=4n$, $\epsilon  = 1/2$ and any $\epsilon_1$ gives that after 
$O(\log n + \log (1/\epsilon_1))= O(b + \log (1/\epsilon_1))$ oracle calls,  Gr{\"o}tschel,  Lov{\'a}sz and Schrijver's  algorithm finds $\bar \bfx=(\bar x_1, \bar x_2) \in K$ satisfying
$|\bar x_2 - E_0(x^*)| \le \epsilon_1.$ 

The new Binary AIFV-$2$  algorithm is  given in \Cref{alg: ellipse2}.

\begin{algorithm}[h]
  \caption{An $O(n^5 b)$ Ellipsoid Based Algorithm for constructing Binary AIFV-$2$ trees}
 \label{alg: ellipse2}
  \begin{algorithmic}[1]
  \State $\epsilon_0  \leftarrow  2^{-2(b+1)}$;  $\epsilon_1 = 2^{-(b+1)} \epsilon_0$ 
 \State Find  $\bar \bfx=(\bar x_1, \bar x_2) \in K$ satisfying  $|\bar x_2 - \hat x_2| \le  \epsilon_1$  
  \State \  
  \State  $l  \leftarrow  \max\left\{0, \bar x_1 - \epsilon_0/2\right\};$   $r \leftarrow  \min\left\{\bar x_1 + \epsilon_0/2,1\right\}$.
    \State Construct   $T_0(l)$, $T_1(l)$,    $T_0(r)$ $T_1(r)$ and their corresponding lines.\\
    \State In $[l,r]$, set \\
    \hspace*{.1in}   $E_0(x) = \min\left(f_{T_0(l)}(x),f_{T_0(r)}(x)\right)$\\
      \hspace*{.1in}   $E_1(x) = \min\left(g_{T_1(l)}(x),f_{T_1(r)}(x)\right)$\\
   \State In $[l,r]$, define $M(x) = \min \left(E_0(x), E_1(x)\right)$
    \State Set  $x'= \argmax_{x \in [l,r]} M(x)$ \quad  \%Break ties arbitrarily
     \State  return  $T_0(x'),\, T_1(x')$.

  \end{algorithmic}
\end{algorithm}

\begin{figure*}[t]
 \center
\includegraphics[width=5.5in]{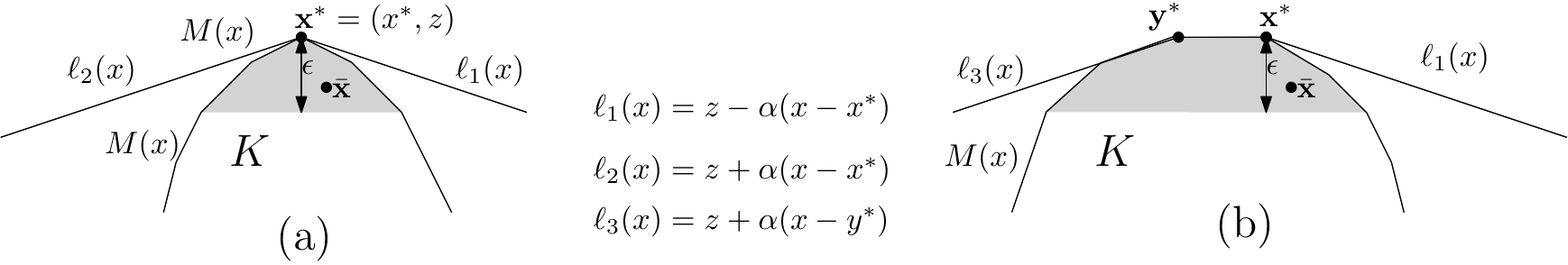}
  \caption{ Illustration of the two scenarios that can occur in Lemma \ref {lem:2Dgeo}.  $M(x)$ is  a bounded piecewise linear concave function and $K$ is the region below it. $z$ is the highest $y$ coordinate in $K.$  The slope of all  supporting lines of $K$ have absolute value $\ge \alpha.$  The ``top'' of $K$ can either be one point (case (a)) or a horizontal line segment (case (b)). In (a),  $K$ must be below the lines $\ell_1(x)$ and $\ell_2(x).$  In (b),  $K$ must be below the lines $\ell_1(x)$, $\ell_3(x)$ and the line $y=z.$ Thus, if $\hat x_2 > z - \epsilon$ for small $\epsilon$,  $\hat \bfx$ must be ``close'' to the top of $K.$
  }
  \label{fig:2Dgeo}
  \end{figure*}

From the discussion above, Line 2 can be implemented via Theorem \ref{thm:ellipsoid} using $O(b + \log (1/\epsilon_1)) = O(b)$ oracle calls.  Since each oracle call requires $O(n^5)$ time, Line 2 requires $O(n^5 b)$ total time.   After Line 2 completes, some point $\bar \bfx$ ``close'' to $\bfx^*$ is known. 

Line 5  can be implemented in $O(n^5)$ time using the DPs. Similar  to the previous section,  Lemma \ref{lem:spacing}  implies that each of  $E_0(x)$ and $E_1(x)$ have at most one critical point in $[l,r]$ and thus  the correctness of lines 8 and 9.  More specifically, each of $E_0(x)$ and $E_1(x)$ comprise only one or two lines in $[l,r]$ and these can be found from
$T_0(l)$, $T_1(l)$,    $T_0(r)$ $T_1(r)$ in $O(1)$ time.  

After completing Lines 5-9, Line 11 is taking the lower envelope of (at most) four lines, so it only requires $O(1)$ time.

Line 12 can then be evaluated in $O(1)$ time.

  Implementing Line 13 takes another $O(n^5)$ time.

Thus the entire algorithm uses $O(n^5 b)$ time.

It remains to prove the correctness of Line 13,  i.e., that $T_0(x'),\, T_1(x')$ are an optimal Binary AIFV-$2$ code.

From Lemma \ref{lem:22cases} it suffices to prove that  there exists $x' \in [l,r]$ with $M(x') = M(x^*).$  The possible situations that can occur are illustrated in \Cref{fig:4cases2D}.  The proof is below.

Recall that the slope of $f_{T_0(z)}$ is  $q_1(T_0(z))$ and the slope of
$g_{T_1(z)}$ is  $-q_0(T_1(z))$. 
The main tool is the observation that because all of the $p_i$ are multiples of $2^{-b},$ these slopes, if not $0,$ 
have absolute value at least $2^{-b}.$

We now use the following basic  geometric lemma:

\begin{lem}
\label{lem:2Dgeo} See \Cref{fig:2Dgeo}.
Let $M(x)$ be  a piecewise linear concave  function such that $z = \max_{x \in \Re}M(x)$ is bounded and let 
$$K = \{(x,y) \,:\,  y \le M(x)\}.$$
Further suppose that 
 there exists $\alpha > 0$ such that,  if  $y = mx +b$ is a supporting line of $K$ with $m \not = 0$,  then  $|m| \ge  \alpha$.
 
 Let $\bar \bfx = (\bar x_1, \bar x_2) \in K.$  Then, if $\bar  x_2 > z-\epsilon,$ there exists $x'$ such that $|\bar x_1 - x'| \le \alpha^{-1} \epsilon$ and  $M(x') = z.$
\end{lem}
\begin{proof}
Let $F = \{(x_1,x_2)  \in K\,:\,  x_2= z.\}$.  Then there are two possible cases.  Either (a) $F$ is one point or (b) $F$ is a horizontal line segment.  

\begin{itemize}
\item[(a)]  $F = \{\bfx^*\}$ for some unique point $\bfx^* = (x^*,z).$

If $\bar x_1 \ge x^*$ then,  by the convexity of $K$ and  the fact that the slope of $M(x)$ everywhere in $(x^*,\infty]$ is at most  $-\alpha$ yields
\begin{eqnarray}
\bar x_2 \le M(\bar x_1)
&\le &   M(x^*) - \alpha  (\bar x_1 - x^*) \nonumber\\
&=& z - \alpha  (\bar x_1 - x^*). \label{eq:right*}
\end{eqnarray}
Thus
$$
 \bar x_1 - x^* \le  \alpha^{-1} \left(z- \bar x_2\right) \le  \alpha^{-1} \epsilon.
$$

If $\bar x_1 \le x^*$ then, again by the convexity of $K$ and  the fact that the slope of $M(x)$ everywhere in $[-\infty,x^*_1]$ is at least $\alpha$ yields
\begin{eqnarray}
\bar x_2 \le M(\bar x_1)
&\le &   M(x^*) - \alpha  (x^*-\bar x_1 ) \nonumber\\
&=& z - \alpha  (x^*-\bar x_1). \label{eq:left*}
\end{eqnarray}
Thus
$$
x^* - \bar x_1  \le  \alpha^{-1} \left(z- \bar x_2\right) \le  \alpha^{-1} \epsilon.
$$

Combining the two directions proves that $|\bar x_1 - x^*| \le \alpha^{-1} \epsilon$, proving the lemma.

\item[(b)]  $F$ is a horizontal line segment.

Let  $\bfy^*=(y^*,z)$ be the leftmost   point in $F$ and $\bfx^*=(x^*,z)$ the rightmost point in $F$. Then 
$$F= \{(x_1,x_2)\,:\,  y^* \le x_1 \le x^* \ \mbox{and} \ x_2=z\}.$$

There are then three cases
\begin{itemize}
\item If $\bar x_1 \ge x^*$ then  Equation \ref{eq:right*} can still be applied, so  $\bar x_1 - x^* \le  \alpha^{-1} \epsilon$ and the lemma is correct with $x' = x_1^*.$

\item If $ y^* \le \bar x_1 \le x^*_1$ then  the  lemma is trivially correct with $x' = \bar x_1$

\item  If $ \bar x_1  < y^*$ then, after replacing  $x^*$ with $y^*$,  Equation \ref{eq:left*} can still be applied.  
Thus  $y^* - \bar x_1 \le  \alpha^{-1} \epsilon$ and the lemma is correct with $x' = y^*.$
\end{itemize}

\end{itemize}

\end{proof}

Plugging into the lemma with  $\alpha = 2^{-b}$ and $\epsilon = \epsilon_1$ says that there exists some $x'$ satisfying $|x' - \bar x_1| \le 2^b \epsilon_1= \epsilon_0/2$ such that
$M(x') = M(x^*)$.  Thus  such an $x' \in [l,r]$ will be found by \Cref {alg: ellipse2}, with correctness following from  Lemma \ref{lem:22cases}.  

\section{Conclusion and Future Directions}
\label{sec:Conc}
 
 This paper gave the first polynomial time algorithms for constructing optimal AIFV-$2$ codes.  The two algorithms presented are only   {\em weakly}-polynomial though, in that they depend upon the number of {\em bits} required to represent  the input.
An obvious open question would be to find an algorithm which, like Huffman coding,  is  strongly polynomial, i.e., only dependent  upon $n$, the number of {\em items} to be encoded.

 Another direction for new research would be to develop polynomial time algorithms for  constructing optimal  Binary AIFV-$m$ codes. 
  These are \cite{introduceMAIFV}
generalizations of  Binary AIFV-$2$ codes that code using  $m$-tuples $T_0,T_1,\ldots,T_{m-1}$ of coding trees.  
They have essentially the same structure as in the $m=2$ case but generalize the rule on how to flip between $m$ different coding trees.  
In addition,   instead of differentiating between
just master and slave nodes, they classify 
  nodes as $k$ different types of master nodes.

 \cite{iterativealgomaifv,iterativeisoptimal} prove that an iterative algorithm similar to \Cref{alg: alg1}  for $m=2$   will also successfully construct an optimal  Binary AIFV-$m$  code for $m >2.$
  Similar to the $m=2$ case, their algorithm  can be shown to be equivalent to associating each of the coding trees in the current $m$-tuple with a hyperplane in $m$-space, finding the intersection point of those hyperplanes, projecting that down to  
 a $(m-1)$-dimensional point $C$ and then constructing a new $m$-tuple of trees as a function of the $C$ (this is a generalization of Lines 4 and 5 in \Cref{alg: alg1}). Again, similar to the $m=2$ case,
  \cite{iterativealgomaifv,iterativeisoptimal} prove that this generalized algorithm converges in a finite number of steps to an optimal Binary AIFV-$m$ code but with no bound on the actual number of such steps.

It is unclear how the simple algorithm of  \Cref{subsec:bsearch} could be extended to $m >2.$
  Preliminary investigations show that  it is very likely that the Ellipsoid Method approach described in \Cref {subsec:ellipsoid} for $m=2$ would work for $m >2$ as well.  
More specifically, in the  $m >2$ case,  the polygon $K$ from Lemma \ref{lem:2Dgeo}  becomes  a {\em polytope} and the problem reduces to finding the {\em highest vertex}  in $K.$  Again,  the Gr{\"o}tschel,  Lov{\'a}sz and Schrijver ellipsoid algorithm will permit finding  a point $\bar \bfx$  {\em close} to the top of $K$ in polynomial time.
If the highest point in $K$ is a unique vertex $\bfx^*$, an analysis similar to the one for $m=2$ will permit using $\bar \bfx$ to find $\bfx^*$ which would then yield an optimal $m$-tuple of coding trees.
This is the non-degenerate case in which  all $m$ trees are actively used for coding messages.

  The complications arise when the top of $K$ is not a unique point but instead a higher-dimensional face.  These cases  would correspond to  degenerate scenarios in which not all $m$ coding  trees are used by the Binary AIFV-$m$ code\footnote{This is a more complicated generalization of the degenerate case when $m=2$ in which Huffman coding is the optimal AIFV-$2$ code, i.e., only $T_0$ is used and $T_1$ is never accessed. More specifically, it is  a generalization of case (b) in the proof of Lemma \ref{lem:22cases}. }.  To prove that knowing $\bar \bfx$ permits reconstructing  a vertex $\bfx^*$ of $K$ (and thus an optimal AIFV-$m$ code)  would require a more nuanced and detailed understanding of the geometry of $K.$

\bibliography{biblio.bib}
\bibliographystyle{ieeetr}


\end{document}